\documentclass[12pt, a4paper, reqno]{amsart}
\usepackage{mathrsfs}
\usepackage{tabu}
\usepackage{bbm}
\usepackage{amsmath,amscd,amssymb,latexsym}
\usepackage[all]{xy}
\usepackage{xcolor}

\input xypic

\evensidemargin=\oddsidemargin
\DeclareMathVersion{can}
\DeclareMathAlphabet{\can}{OT1}{cmss}{m}{n}
\newtheorem{thm}{Theorem}[section]
\newtheorem{cor}[thm]{Corollary}

\newtheorem{prop}[thm]{Proposition}

\newtheorem{exa}[thm]{Example}
\theoremstyle{definition}

\theoremstyle{fact}

\theoremstyle{conjecture}

\numberwithin{equation}{section}
\newcommand{\Tr}{\operatorname{Tr}}

\newcommand{\mwang}[1]{{{\color{purple}#1}}}

\begin{document}
%\largewo
\title[] { Further improvement on index bounds}
\author[Y. Wu]{Yansheng Wu}
\address{\rm School of Computer Science, Nanjing University of Posts and Telecommunications, Nanjing 210023, P. R. China;  Shanghai Key Laboratory of Trustworthy Computing, East China Normal University, Shanghai 200062, P. R. China}
 \email{yanshengwu@njupt.edu.cn}
%\author[Q. Yue]{Qin Yue}
%\address{\rm Department of Mathematics, Nanjing University of Aeronautics and Astronautics,
%Nanjing,  211100, P. R. China}
%\email{yueqin@nuaa.edu.cn}

\author[Y. Lee]{ Yoonjin Lee}
\address{\rm Department of Mathematics, Ewha Womans
University, Seoul, 03760, Republic of Korea}
\email{yoonjinl@ewha.ac.kr}

 \author[Q. Wang]{Qiang Wang}
 \address{\rm School of Mathematics and Statistics, Carleton University, Ottawa, ON K1S 5B6, Canada}
 \email{wang@math.carleton.ca}
 %% \author[S. Yang]{Shudi Yang}
% \address{\rm School of Mathematical Sciences, Qufu Normal University, Qufu, 273165, P.R.China}
% \email{yangshd3@mail2.sysu.edu.cn}

%\thanks{The paper is supported by }


% ----------------------------------------------------------------

\thanks{
{\tiny Y. Wu was supported by the National Natural Science Foundation of China (Grant No. 12101326), the Natural Science Foundation of Jiangsu Province (Grant No. BK20210575 ).   
 Y. Lee is a corresponding author and supported by
the National Research Foundation of Korea (NRF) grant funded by the Korea government (MEST)(NRF-2017R1A2B2004574) and also by the Basic Science Research Program through the National Research Foundation of Korea (NRF) funded by the Ministry of Education (Grant No. 2019R1A6A1A11051177). The research of Qiang Wang is partially supported by NSERC of Canada (RGPIN-2017-06410). }}

\begin{abstract}  
In this paper we obtain further improvement of index bounds for character sums of polynomials over finite fields. We present some examples, which show that our new bound is  an improved bound compared to both the Weil bound and the index bound given by Wan and Wang.
As an application, we count the number of all the solutions of some algebraic curves by using our result. 
%As a method, we consider an equivalence relation of polynomials over finite fields, and we show that all polynomials in an equivalence class have the same Weil sum. }
%{\color{red} By replacing the polynomial, we reduce the index bound of the new polynomial.

%Thanks to an important result on coset leaders by Ding (IEEE Trans. Inf. Theory,  61(10):   5322-5330, 2015), we present a unified improvement of the Weil bound for character sums over finite fields. 

\end{abstract}
\maketitle
%\tableofcontents
% ----------------------------------------------------------------
%\setcounter{section}{-1}

\section{Introduction}

Let  $p$  be  a prime number and  $ \Bbb F_q$ be a finite field, where $q=p^m$ for some positive integer $m$.  Let  $\psi: \Bbb F_q \to \mathbb{C} ^*$ be a nontrivial additive character and  $f(x)\in \Bbb F_q[x]$. The following sums
$$\sum_{x\in \Bbb F_q}\psi(f(x))$$
are referred to as the {\it Weil sums} \cite{LN}. Throughout this paper,  we always view a polynomial $f(x)\in\Bbb F_q[x] $ as a mapping over $\Bbb F_q$.  In this sense the degree of a polynomial in $\Bbb F_q[x]$ should be controlled  by $q-1$ when we consider  Weil sums.

Given a polynomial $f(x)\in\Bbb F_q[x] $ of positive degree $n$ with  $\gcd(n,q)=1$, we have
\begin{equation}\bigg|\sum_{x\in \Bbb F_q}\psi(f(x))\bigg|\le (n-1)\sqrt{q}.\end{equation}
 The upper bound in  (1.1) is well known as the {\it Weil bound}.

The Weil bound for character sums has many applications  in mathematics, theoretical computer science, and information theory etc.  
The Weil bound is trivial when the degree of the polynomial  is bigger than $\sqrt{q}$. 
Some progress on improvement to the Weil bound has been made as follows. 

(1) Using {\it Deligne's bound} for exponential sums in several variables, Gillot et al. \cite{G, GL} provided a new bound for $|\sum_{x\in \Bbb F_{q^m}}\psi(f(x))|$ such that for any $b\in \Bbb F_{q^m}^*$ the $q$-degree of  $f(x)=bx^d+g(x)\in \Bbb F_{q^m}[x]$ only depends on the leading term $bx^d$ {\color{orange};  the {\it $q$-degree} of a nonzero polynomial $f$ over $\Bbb F_{q^m}$ is defined by \begin{equation*}\mbox{deg}_{q}(f)=\underset{}{\mbox{max}}\{d_{0}+d_{1}+\cdots+d_{m-1}:  d\in \mbox{supp}(f)\}, \end{equation*} where  $d=d_{0}+d_{1}q+\cdots+d_{m-1}q^{m-1}$ denotes the $q$-ary expansion of $d$. }

(2) Wan and Wang \cite{WW} obtained an index bound for character sums over finite fields; they used the the concept of index of a polynomial over a finite field, which was first introduced in the research of permutation polynomials \cite{AGW}. This bound improved the Weil bound for high degree polynomials with small indices as well as polynomials with large indices that are generated by cyclotomic mapping of small indices. 

(3) Recently, there is an improvement on the Hasse-Weil bound in the characteristic two case by Cramer and Xing \cite{CX}. 
They used the algebraic geometry and the algebraic number theory. 
As an application, they also improved the Weil bound for character sums.

\subsection{Main result and  comparison with previous results} 
\phantom{.}\\

The concept of the {\it index} of a polynomial over a finite field was first introduced in \cite{AGW}.

A polynomial $f(x)\in  \Bbb F_{q}[x] $ of degree $n\le q-1$ can be written as the following form:
$$f(x)=a(x^n+a_{n-i_1}x^{n-i_1}+\cdots+a_{n-i_k}x^{n-i_k})+b,$$
where $a, a_{n-i_j}\neq 0, j=1,\ldots,k$. Let $r$ be the lowest degree of $x$ in $f(x)-b$. The index $ l$ of the polynomial $f(x)$ is defined by 
$$l:=\frac{q-1}{\gcd(n-r,n-r-i_1, \ldots, n-r-i_{k-1},q-1)}.$$ In fact, any non-constant polynomial $f(x)\in  \Bbb F_{q}[x] $ of degree $n \le q-1$ can be written uniquely as $f(x)=a(x^r h(x^{(q-1)/l}))+b$.  

\

{\color{blue} By using the index of polynomials, Wan and Wang \cite{WW} found a new bound, which is called an {\it index bound}, for character sums of polynomials over finite fields as follows:

\begin{thm} \cite[Theorem 1.1]{WW} {\rm
{Let $f(x)=x^r h(x^{(q-1)/l})+b\in \Bbb F_q[x]$ be any polynomial with index $l$. 
Let $\xi$ be a primitive $l$-th root of unity in $\Bbb F_q$ and $n_0=\sharp\{0\le i\le l-1 \mid h(\xi^i)=0\}$. }
Let $\psi: \Bbb F_q \to \mathbb{C} ^*$ be a nontrivial additive character.  
Then 
\begin{equation}\bigg|\sum_{x\in \Bbb F_q}\psi(f(x))-\frac{q}{l}n_0\bigg|\le (l-n_0)\gcd\bigg(r, \; \frac{q-1}{l}\bigg)\sqrt{q}.\end{equation}
 %Let $q$ be a  power of a prime number $p$ and $f(x)=x^rg(x^{(q-1)/l})+b\in \Bbb F_q[x]$ be any polynomial with index $l$. Let $\xi$ be a primitive $l$th root of unity in $\Bbb F_q$ and $n_0=\sharp\{0\le i\le l-1 | f(\xi^i)=0\}$. Let $\psi: \Bbb F_q \to \mathbb{C} ^*$ be a nontrivial additive character.  Then 
%\begin{equation}\bigg|\sum_{x\in \Bbb F_q}\psi(f(x))-\frac{q}{l}n_0\bigg|\le (l-n_0)\gcd\bigg(r,\frac{q-1}{l}\bigg)\sqrt{q}.\end{equation}

% \begin{lem} \label{lem4}\cite[Theorem 1.1]{WW} {\rm

 }
 \end{thm}}
 
 \mwang{This theorem is sensitive to the form of $f(x)$. Replacing $f(x)$ with $f^{*} (x) = f(x) + d(x)$
  such that  $\sum_{x\in \Bbb F_q}  \psi(d(x)) = 0$, and in particular, for $d(x)$ of the form $y^p -y$, leaves the
  left hand side  unchanged but may lower the upper bound in the right hand side.}

As mentioned before, we can view  a polynomial as a mapping over $\Bbb F_{q} $ when we just consider the Weil sums. Then it is easy to extend the Frobenius mapping $\pi$ over $\Bbb F_q$ to $\Bbb F_{q}[x]$ as follows:
$$\pi : \Bbb F_{q}[x] \to \Bbb F_{q}[x], \pi(a) = a^p, \pi(x)=x^p.$$

There is an equivalence relation between polynomials  in $\Bbb F_{q}[x]$. Given two polynomials $f(x)=\sum_ia_ix^i$ and $g(x)$ in $\Bbb F_{q}[x]$, we write $f\sim g$ if there exists a vector $v=(v_0, v_1,\ldots, v_n)\in \mathbb{Z}_m^{n+1}$ such that $$\pi^v(f)=\pi^{v_0}(a_0)+\pi^{v_1}(a_1x)+\cdots+\pi^{v_n}(a_nx^n)=g.$$ The equivalence class of $f(x) \in \mathbb{F}_q[x]$ is denoted by $[f]$.
For each polynomial $g(x)$ in the equivalence class $[f]$, we  use $l_g$ for  the index of $g(x)$ and $n_g$ for the degree of $g(x)$. For each equivalence class $[f]$, let $l^*$ be defined by  $$l^*=\mbox{min}\{l_g: g\in [f]\}.$$ Then there exists a polynomial $f^*(x)\in [f]$ such that 
$f^*(x)=x^{r^*} h^*(x^{(q-1)/l^*})+b\in \Bbb F_q[x]$ and
the index of  $f^*(x)$ is exactly $l^*.$    We call  $f^*(x)$ the equivalent polynomial of $f(x)$ with the smallest index and $h^*(x)$ the associated polynomial of $f^*(x)$.

\

The following theorem is our main result, which provides a \mwang{new formulation} of index bounds in Theorem 1.1 for character sums of polynomials over finite fields. 

\begin{thm} \label{thm1.1}
Let $f(x)\in \Bbb F_q[x]$ be a polynomial of positive degree $n$ with $\gcd(n,q)=1$. {Let $f^*(x)$ be  the equivalent polynomial of $f(x)$ with the smallest index $l^*$ and $h^*(x)$ the associated polynomial of $f^*(x)$.  Let $\xi$ be a primitive $l^*$-th root of unity in $\Bbb F_q$ and $n_0=\sharp\{0\le i\le l^*-1 \mid  h^*(\xi^i)=0\}$.} Let $\psi: \Bbb F_q \to \mathbb{C} ^*$ be a nontrivial additive character.  Then 
\begin{equation}\label{indexbound}
\bigg|\sum_{x\in \Bbb F_q}\psi(f(x))-\frac{q}{l^*}n_0\bigg|\le (l^*-n_0)\gcd\bigg(r^*, \; \frac{q-1}{l^*}\bigg)\sqrt{q}.\end{equation}

\end{thm}

\mwang{The following examples show that the upper bound in \eqref{indexbound}  is indeed  an improved bound compared to both the Weil bound in (1.1) and the index bound in Theorem 1.1. }

 \begin{exa} {\rm 
  Let $f(x)=x^{25}+ax^4\in \Bbb F_{27}[x]$, where $a\in \Bbb F_{27}^*$.   {Obviously, the Weil bound is trivial because of the high degree. Since the index of $f$ is $26$, the index bound is also trivial.}
  % Note that $C_4=\{4,10,12\}$  is a $3$-cyclotomic coset modulo $26$. In \cite[Corollary 1.2]{WW},
%$\bigg|\sum_{x\in \Bbb F_q}\psi(x^{25}+ax^4)-\frac{27}{26}\bigg|\le 25\sqrt{27}$; 
%otherwise, $\bigg|\sum_{x\in \mathbb{F}_q}\psi(x^{25}+ax^4) \bigg|\le26\sqrt{27}.$ 
However, by Theorem \ref{thm1.1}, $l^*=\frac{26}{\gcd(25-12,26)}=2$. 
%If $x^{13}+a^3$ has a solution in the subset of all $2$nd roots of unity in $\Bbb F_{27}$,
If $a^3  = \pm 1$, 
 then 
$\bigg|\sum_{x\in \Bbb F_{27}}\psi(x^{25}+ax^4)-\frac{27}{2}\bigg|\le \sqrt{27}$;
otherwise, we have $\bigg|\sum_{x\in \Bbb F_{27}}\psi(x^{25}+ax^4) \bigg|\le2\sqrt{27}.$ %However, by \eqref{eqn:n*}, we have  $n^*=17$ and the upper bound in \eqref{eqn:reducedweil} is trivial.

%(2) Let $f(x)=x^{17}+ax^4\in \Bbb F_{27}[x]$.

 }
 \end{exa}

  \begin{exa} {\rm 
  Let $f(x)=x^{19}+ax^2\in \Bbb F_{27}[x]$, where $a\in \Bbb F_{27}^*$.   Obviously, the Weil bound and the index bound are both  trivial.
By Theorem \ref{thm1.1}, $l^*=2$.  If \mwang{ $a^3  = \pm 1$, }
 then 
$\bigg|\sum_{x\in \Bbb F_{27}}\psi(x^{19}+ax^2)-\frac{27}{2}\bigg|\le \sqrt{27}$;
otherwise, we have $\bigg|\sum_{x\in \Bbb F_{27}}\psi(x^{19}+ax^2) \bigg|\le2\sqrt{27}.$

%(2) Let $f(x)=x^{17}+ax^4\in \Bbb F_{27}[x]$.
}
 
 \end{exa}

 \vskip 0.5cm

Similarly, we define
\begin{equation}\label{eqn:n*}
n^* = \mbox{min}\{n_g: g\in [f]\}.
\end{equation}
Then using the Weil bound,  the same arguments derives
\begin{equation}\label{eqn:reducedweil}
\bigg|\sum_{x\in \Bbb F_q}\psi(f(x)) \bigg| \le (n^*-1)\sqrt{q}.
\end{equation}

The following examples show that the upper bound in \eqref{eqn:reducedweil}  is also an \mwang{improved}  bound {\color{orange} when it is}  compared to  the Weil bound in (1.1),  the index bound in Theorem 1.1, and the upper bound in \eqref{indexbound}.
%\subsection{Our techniques} 

  \begin{exa} {\rm 
  Let $f(x)=x^{19}+ax^4\in \Bbb F_{27}[x]$, where $a\in \Bbb F_{27}^*$.   Obviously, the Weil bound, the index bound, and the upper bound in \eqref{indexbound} are all trivial.  However, by \eqref{eqn:n*} and \eqref{eqn:reducedweil}, we have  $n^*=5$ and 
 $\bigg|\sum_{x\in \Bbb F_{27}}\psi(x^{19}+ax^4) \bigg|\le4\sqrt{27}.$

%(2) Let $f(x)=x^{17}+ax^4\in \Bbb F_{27}[x]$.

 }
 \end{exa}

   \begin{exa} {\rm 
  Let $f(x)=x^{10}+ax^5\in \Bbb F_{27}[x]$, where $a\in \Bbb F_{27}^*$.   Obviously, the Weil bound,   the index bound, and  the upper bound in \eqref{indexbound} are all  trivial.  However, by \eqref{eqn:n*} and \eqref{eqn:reducedweil}, we have  $n^*=5$ and 
 $\bigg|\sum_{x\in \Bbb F_{27}}\psi(x^{10}+ax^5) \bigg|\le4\sqrt{27}.$

%(2) Let $f(x)=x^{17}+ax^4\in \Bbb F_{27}[x]$.

 }
 \end{exa}
 
The $p$-cyclotomic coset modulo $q-1$ containing $i$ is defined by $$C_i =\{ip^j \pmod {q-1}: 0\le j <l_i\},$$
where $l_i$ is the smallest positive integer such that $p^{l_i} i\equiv i \pmod {q-1}$, and is the cardinality of $C_i$. The smallest integer in $C_i$ is called the {\it coset leader} of $C_i$. By \cite[Lemma 6]{L}, we also find an infinite family of polynomials, which shows that the upper bound  in \eqref{eqn:reducedweil} is  \mwang{an improved bound}. 
 
\begin{prop} {\rm Let $q=p^m$,  $s$ be an integer with $p<s<\sqrt q$,  and $\gcd(s,p)=1$.  Let $f(x)=x^{sp^{m-1}}+ax^r\in \Bbb F_q[x]$, where $a\in \Bbb F_q^*$ and $r<s$. Then $\gcd(sp^{m-1}-q+1,q)=1$, $n^*=s$, and  $\bigg|\sum_{x\in \Bbb F_q}\psi(f(x)) \bigg| \le (s-1)\sqrt{q}$.

}
\end{prop}

%However, we point out that there are some cases where index bounds for polynomials over the finite field $\Bbb F_q$ are impossible to be improved {\color{red} using this approach}; for example, the case that $q-1$ is a prime number.}

 %some cosets and the property of trace functions. 
 
 % In order to achieve our goal, the first step is to decompose the Weil sum of a polynomial into the Weil sums of many monomials, and then we replace these monomials {\color{red}(by the remaining Weil sums ??)}. 
% The second step is to make sure that the index of the {\color{red}new} polynomial {\color{red}(obtained by replacing their monomials with ??)} is smaller than {\color{red}the initial polynomial.}
 %In this way, we can find a new polynomial with better bound; the important point is that the Weil sums of two polynomials {\color{red} (the initial one and the new one)} remain unchanged. 

%However, we point out that there are some cases where index bounds for polynomials over the finite field $\Bbb F_q$ are impossible to be improved {\color{red} using this approach}; for example, the case that $q-1$ is a prime number.}

%  
%\subsection{Organization of the paper}

%\

%In this paper, we focus on further improvement of index bounds for character sums over finite fields.  
The rest of this paper is organized as follows. 
In Section 2, we prove Theorem ~\ref{thm1.1}. 
Our main idea of improving the index bound in \cite{WW} is to replace the polynomial over finite fields by using 
the equivalent polynomial with the smallest index. Those equivalent polynomials are obtained  by applying Frobenius automorphisms on individual monomials.   In Section 3, as an application, we count the number of solutions of some algebraic curves by using our main result.

 \section{Proof of Theorem 1.2 and corollaries}

Let $q=p^m$ as before and  $\Bbb F_p$ be a subfield of $\Bbb F_{q}$. 
Let $\pi$ be the {\it Frobenius automorphism} of $\Bbb F_{q}$ over $\Bbb F_p$, which is defined by $\pi : \Bbb F_{q} \to \Bbb F_{q} , \pi(a) = a^p$.
We note that $\pi(a) = a$ if and only if $a \in \Bbb F_p$. 

 \vskip 0.3cm
 
Now, we are ready to give the proof of Theorem 1.2.

\vskip 0.3cm

{\bf Proof of Theorem 1.2.} We write $f(x)=a_0+a_1x+\cdots+c_nx^n$ and $\psi=\psi_1(c)$ for some nonzero $c\in \Bbb F_q$, where $\psi_1$ is the {\it canonical} additive character of $\Bbb F_q$. 
Let $\Tr_{q/p}$ be the trace function from $\Bbb{F}_q$ to $\Bbb{F}_p$.
For each $v=(v_0, v_1, \ldots, v_n)\in \mathbb{Z}_m^{n+1}$,  
\begin{eqnarray*} && \sum_{x\in \Bbb F_q}\psi_1(\pi^v(cf(x)))\\
&=& \sum_{x\in \Bbb F_q}\psi_1 (\pi^{v_0}(ca_0)+\pi^{v_1}(ca_1x)+\cdots+\pi^{v_n}(ca_nx^n))\\
&=& \sum_{x\in \Bbb F_q}\psi_1 ((ca_0)^{p^{v_0}}+(ca_1x)^{p^{v_1}}+\cdots+(ca_nx^n)^{p^{v_n}})\\
&=& \sum_{x\in \Bbb F_q}\zeta_p^{\Tr_{q/p}( (ca_0)^{p^{v_0}}+(ca_1x)^{p^{v_1}}+\cdots+(ca_nx^n)^{p^{v_n}})}\\
&=& \sum_{x\in \Bbb F_q}\zeta_p^{\Tr_{q/p}( ca_0+ca_1x+\cdots+ca_nx^n)}\\
&=&\sum_{x\in \Bbb F_q}\psi(f(x)).
\end{eqnarray*}

Note that $f(x)$ and $cf(x)$ have the same indices. 
Taking a vector $v\in  \mathbb{Z}_m^{n+1}$ such that $f^*(x)=\pi^v(cf(x))$, the result follows from Theorem 1.1. \hfill$\square$

%The {\it $p$-cyclotomic coset} modulo $q-1$ containing $i$ is defined by $$C_i =\{ip^j \pmod {q-1}: 0\le j <l_i\},$$
%where $l_i$ is the smallest positive integer such that $p^{l_i} i\equiv i \pmod {q-1}$, and $l_i$ is the cardinality of $C_i$.  

\vskip 0.3cm

Using the concept of the $p$-cyclotomic cosets, for binomial polynomials, we obtain the following corollary, which is a simple version of Theorem 1.2. This also improves and corrects an error in Corollary 1.2 in \cite{WW}. 

%Let $\alpha$ be a primitive $n$th root of unity over some extend field of $\Bbb F_q$. It is easily seen that the minimal polynomial $m_i(x)$ of $\alpha^i$ over $\Bbb F_q$ is given by $$m_i(x)=\prod_{j\in C_i}  (x-\alpha^j).$$

 \begin{cor} {\rm Let   $f(x)=x^{n}+ax^{r}\in \Bbb F_{q}[x]$, where $a\in \Bbb F_{q}^*$ and $1\le r<n\le q-1$.  Let $C_r$ be  the $p$-cyclotomic coset modulo $q-1$ containing $r$. 
Suppose that $r^*=rp^k$ for some integer $k$ with $0\le k\le m-1$ such that  $$l^*=
\frac{q-1}{\gcd(n-r^*,q-1)} =
\mbox{min}\bigg\{\frac{q-1}{\gcd(n-j,q-1)}: j\in C_r\bigg\}.$$
 Let $t=\gcd(n,r,q-1)$.  Let $\psi: \Bbb F_q \to \mathbb{C} ^*$ be a nontrivial additive character. {If $x^{n-r^*}+a^{p^k}$ has a solution in  $\Bbb F_q^*$}, then 

\begin{equation}\bigg|\sum_{x\in \Bbb F_q}\psi(x^{n}+ax^{r})-\frac{q}{l^*}\bigg|\le (l^*-1)t\sqrt{q};\end{equation}
otherwise, we have  
\begin{equation}\bigg|\sum_{x\in \Bbb F_q}\psi(x^{n}+ax^{r}) \bigg|\le l^*t\sqrt{q}.\end{equation} }
 \end{cor}

 \begin{proof}  
 By Theorem 1.2, we have  $$l^*=\mbox{min}\bigg\{\frac{q-1}{\gcd(i-j,q-1)}: i\in C_n, j\in C_r\bigg\}.$$ Suppose that $i=np^{\alpha}$ and $j=rp^{\beta}$ such that $l^*=\frac{q-1}{\gcd(i-j,q-1)}.$ Then $i-j=p^{\alpha}(n-rp^{\beta-\alpha})$; hence, using Theorem 1.2, the result follows immediately since $\gcd(p,q-1)=1$. In particular, $n_0 =1$ if there exists a solution for $x^{n-r^*} + a^{p^k}$. 
 \end{proof}

 %\begin{proof} Note that $\psi(x^n+ax^r)=\psi(x^n)\psi(ax^r)=\psi(x^n)\psi((ax^r)^{p^k})$ for   some positive integer $k$ with $0\le k\le m-1$.
 
 % Let $r^*\in C_r$ such that $l^*=\mbox{min}\bigg\{\frac{q-1}{\gcd(n-b,q-1)}: b\in C_r\bigg\}=\frac{q-1}{\gcd(n-r^*,q-1)} .$  It is easy to verify that there exists $a^*\in \Bbb F_q$ such that $\psi(x^n+ax^r)=\psi(x^n+a^*x^{r^*})$ for each $x\in \Bbb F_q$.
 % Therefore it is enough to consider the new polynomial $x^n+a^*x^{r^*}$. The results follow from Eqs. (3.1) and (3.2).
  %\end{proof}

 In the following corollary, we present some special classes of polynomials for illustration of our main result.

 \begin{cor}  {\rm Let $Q$ be the least prime factor of $q-1$,  $n=\frac{q-1}{Q}$, and $\alpha$ be a positive integer. 
 Let $f(x)=x^{n+p^{\alpha}}+ax\in \Bbb F_{q}[x]$, where $a\in \Bbb F_{q}^*$  and $1< n+p^{\alpha}\le q-1$.   
{If $x^{n}+a^{p^{\alpha}}$ has a solution  in $\Bbb F_q^*$}, then 
\begin{equation*}\bigg|\sum_{x\in \Bbb F_q}\psi(x^{n+p^{\alpha}}+ax)-\frac{q}{Q}\bigg|\le (Q-1)\sqrt{q};\end{equation*}
otherwise, we have 
\begin{equation*}\bigg|\sum_{x\in \Bbb F_q}\psi(x^{n+p^{\alpha}}+ax) \bigg|\le Q\sqrt{q}.\end{equation*}

 }

 \end{cor}
 
 \begin{proof}   For each integer $i$ with $0\le i\le m-1$, $\gcd(n+p^{\alpha}-p^i,q-1)$ is a divisor of $q-1$. 
  Since $Q$ is the least prime factor of $q-1$, $\gcd(n,q-1)$ is the largest integer in the set $\big\{\gcd(n+p^{\alpha}-p^i,q-1):i=0,1,\ldots,m-1\big\}$. Hence, we have $l^*=Q$; thus, the result follows  from Corollary 2.1.
 \end{proof}

%  \begin{rem}{\rm It is easy to verify that the upper bounds in Corollary  2.4 are better than the Weil bound in (1.1) and the index bounds in \cite[Corollary 1.2]{WW}. 
  %In fact, the index of the polynomial $x^{n+p^{\alpha}}+ax$ is $l=\frac{q-1}{\gcd(n+p^{\alpha}-1,q-1)}$. 
 % We note that $\gcd(n+p^{\alpha}-1,q-1)=\gcd(n+p^{\alpha}-1,(p^{\alpha}-1)Q)$ and  there exists a prime number $T$ such that  $v_T (q-1)>v_T(p^{\alpha}-1)$, {\color{blue}where $v_T$ denotes the $T$-adic valuation.}
  %Therefore, we get $l^*=Q<l$.
 
%  }
 %\end{rem}

 Let  $n = p_{1}^{\alpha_{1}}p_{2}^{\alpha_{2}} \cdots p_{h}^{\alpha_{h}}$ be the prime factorization, where $p_1,\ldots, p_h$ are distinct primes and  $\alpha_i$ are positive integers for $1\le i\le h$.
 We denote $\mbox{rad}(n) = p_{1} p_{2} \cdots p_{h}$ and  $v_{p_i}(n)=\alpha_i$ for $1\le i\le h$,  {where $v_T$ denotes the $T$-adic valuation.}

 \begin{cor} {\rm 
 Let $p$ be an odd prime number and $n$ be  an odd positive integer such that $\gcd(n,p-1)=1$ and $2\mbox{rad}(n)=\mbox{rad}(p+1)$. 
 Let    $f(x)=x^{n+p}+a^{}x^{}\in \Bbb F_{p^2}[x]$, where $a\in \Bbb F_{p^2}^*$  and $1< n+p \le p^2-1$.   
 Let $s=\frac{p^2-1}{\gcd(n,p+1)}$.  {If $x^{n}+a^{p}$ has a solution in $\Bbb F_{p^2}$}, then 
\begin{equation*}\bigg|\sum_{x\in \Bbb F_{p^2}}\psi(x^{n+p}+ax)-\frac{p^2}{s}\bigg|\le (s-1)p;
\end{equation*}
otherwise, we have 
\begin{equation*}\bigg|\sum_{x\in \Bbb F_{p^2}}\psi(x^{n+p}+ax) \bigg|\le sp.\end{equation*}
 }
 \end{cor}
 
 \begin{proof} We note that $\gcd(n+p-1,p^2-1)=1$ and $\gcd(n,p^2-1)=\gcd(n,p+1)>1$. 
 Hence, we have $l^*=s<l=p^2-1$. 
 The result thus follows immediately from Corollary 2.1.
   \end{proof}

 \begin{exa}
 {\rm  We present more polynomials $x^n+ax^r\in \Bbb F_{p^m}[x]$ where $a\in \Bbb F_{p^m}^*$ to illustrate our main result for  small primes $p\le 5$ as listed in Tables 1-3.
In the tables, $\ast$ indicates that the bound {is trivial}.  %{Questions:  conditions on $a$? Can we apply Gillot et al's result?  For $p=2$, it is possible to use Cramer and Xing's result?  How about our results comparing to these results?}
% is greater than the size of the field $\Bbb F_{p^m}$; that is, the bound with $\ast$ is a trivial bound. 
 }
 
 \end{exa}
 \begin{table} [h]
\caption{ $p=2$}  
\begin{tabu} to 1\textwidth{X[1,c]|X[1,c]|X[1,c]|X[1.5,c]|X[1.6,c]|X[1.5,c]} 
\hline 
$m$&$n$&$r$&Weil bound& Index bound&Our bound\\ 
\hline \hline 
$4$&$13$&$4$&${\ast}$& ${\ast}$&12\\ 
\hline 
$6$&$41$&$5$&${\ast}$& 56&24\\ 
\hline
$6$&$43$&$25$&${\ast}$& 56&24\\ 
\hline
$6$&$53$&$25$&${\ast}$& ${\ast}$&24\\ 
\hline
$8$&$57$&$12$&${\ast}$&$ {\ast}$&80\\ 
\hline
$8$&$63$&$3$&${\ast}$&$ {\ast}$&80\\ 
\hline
\end{tabu}  
\end{table}

 \begin{table} [h]
\caption{$p=3$}  
\begin{tabu} to 1\textwidth{X[1,c]|X[1,c]|X[1,c]|X[1.5,c]|X[1.6,c]|X[1.5,c]} 
\hline 
$m$&$n$&$r$&Weil bound& Index bound&Our bound\\ 
\hline \hline 
$2$&$7$&$1$&${\ast}$& ${\ast}$&6\\ 
\hline
$3$&$19$&$2$&$ {\ast}$& ${\ast}$&6$\sqrt 3$\\ 
\hline
$4$&$44$&$28$&${\ast}$& 45&18\\ 
\hline
$4$&$46$&$18$&${\ast}$&${\ast}$&36\\ 
\hline
$5$&$154$&$11$& ${\ast}$& ${\ast}$&18$\sqrt 3$\\ 
\hline
$6$&$107$&$9$& ${\ast}$&${\ast}$&189\\ 
\hline
$6$&$122$&$18$& ${\ast}$ & 378&108\\ 
\hline 
\end{tabu}  
 
\end{table}

 \begin{table} [h]
\caption{$p=5$}  
\begin{tabu} to 1\textwidth{X[1,c]|X[1,c]|X[1,c]|X[1.5,c]|X[1.6,c]|X[1.5,c]} 
\hline 
$m$&$n$&$r$&Weil bound& Index bound&Our bound\\ 
\hline \hline 
$2$&$14$&$10$&${\ast}$& ${\ast}$&20\\ 
\hline
$2$&$19$&$11$&${\ast}$& 15&10\\ 
\hline
$3$&$33$&$10$& ${\ast}$& ${\ast}$&20$\sqrt5$\\ 
\hline
$3$&$77$&$3$& ${\ast}$& ${\ast}$&10$\sqrt5$\\ 
\hline
$4$&$42$&$10$&${\ast}$&${\ast}$&150\\ 
\hline
$4$&$314$&$50$&${\ast}$&${\ast}$&100\\ 
\hline 

\end{tabu}  
 
\end{table}

\section{An Application}

In this section, as an application of Theorem 1.2, {\color{blue} we give an estimation of} the number of solutions  for  some algebraic curves.

%\subsection{Numbers of solutions of Artin-Schreier curves}

%\

Let $f(x)\in \Bbb F_{q^m}[x]$ be a polynomial and $N_{f,q^m}$ be the number of solutions $(x,y)\in \Bbb F_{q^m}^2$ of an {\it Artin-Schreier} equation $y^q-y=f(x)$. Then $$N_{f,q^m}=\sum_{\psi_m}\sum_{x\in \Bbb F_{q^m}}\psi_m(f(x)),$$
where the outer sum runs over all additive characters $\psi$ of $\Bbb F_q$ and $\psi_m(x)=\psi(\Tr _{q^m/q}(x))$.

If $f(x)$ has degree $n$ and $\gcd(n,q)=1$, then we have the well known Weil bound:
\begin{equation} \label{WeilCurveBound}
\bigg|N_{f,q^m}-q^m\bigg|\le (q-1)(n-1)\sqrt{q^m}.
\end{equation}

\vskip 0.3cm
By Theorem 1.2, we have the following corollary.

% Let $\xi$ be a primitive $l^*$th root of unity in $\Bbb F_q$ and $n_0=\sharp\{0\le i\le l^*-1 | f^*(\xi^i)=0\}$. Let $\psi: \Bbb F_q \to \mathbb{C} ^*$ be a nontrivial additive character.  Then 
%\begin{equation}\bigg|\sum_{x\in \Bbb F_q}\psi(f(x))-\frac{q}{l^*}n_0\bigg|\le (l^*-n_0)\gcd\bigg(r^*,\frac{q-1}{l^*}\bigg)\sqrt{q}.\end{equation}

\begin{cor} {\rm Let $f(x)\in \Bbb F_{q^m}[x]$ be a polynomial of degree $n$ with $\gcd(n,q)=1$. Let $N_{f,q^m}$ be the number of solutions $(x,y)\in \Bbb F_{q^m}^2$ of an Artin-Schreier equation $y^q-y=f(x)$.  Then we get
   \begin{equation*} 
\bigg|N_{f,q^m}-q^m-\frac{(q-1)q^mn_0}{l^*}\bigg|\le  (q-1)(l^*-n_0)\gcd\bigg(r^*,\frac{q^m-1}{l^*}\bigg)\sqrt{q^m}.
\end{equation*}
  %except the case when $x^{n-r^*}+a^{p^k}$ has a root in the set of $l^*$th root of unity in $\Bbb F_{q^m}$, in which case 
 %  \begin{equation*} 
%\bigg|N_{f,q^m}-q^m-\frac{(q-1)q^m}{l^*}\bigg|\le (q-1)(l^*-1)t\sqrt{q^m}.
%\end{equation*}

}
\end{cor}

\begin{exa} {\rm Let $f(x)=x^{13}+ax\in \Bbb F_{16^2}[x]$ with $a\in \Bbb F_{16^2}^*$. Let $N_{f,16^2}$ be the number of solutions $(x,y)\in \Bbb F_{16^2}^2$ of an Artin-Schreier equation $y^{16}-y=f(x)$. 
%By \cite[Corollary 2.2]{WW} we get \begin{equation*} 
%\bigg|N_{f,16^2}-16^2\bigg|\le  15\cdot 255\cdot16.
%\end{equation*}  
Note that the $2$-cyclotomic coset $C_{13}$ modulo 255 is given by $C_{13}=\{13, 26, 52, 67, 104,134\}$.
By Corollary 3.1,  $l^*=\frac{255}{\gcd(52-1,255)}=5. $ Then we get
\begin{equation*} 
\bigg|N_{f,16^2}-16^2\bigg|\le  15\cdot 5\cdot 16,
\end{equation*} except  {the case when $x^{51}+a$ has a solution in  $\Bbb F_{16^2}$}, in which
case, we have
\begin{equation*} 
\bigg|N_{f,16^2}-16^2-\frac{15\cdot16}{5}\bigg|\le  15\cdot 4\cdot 16.
\end{equation*} 
By Magma program, $N_{f,16^2}=1024$ if $a^5=1$ and $N_{f,16^2}=256$  otherwise.  
Despite that the above bounds are not close to the reality for $a^5 \neq 1$,   our bounds are still an improvement of  \eqref{WeilCurveBound} and the above second bound is pretty good when $a^5=1$.

%\mwang{Question: Can we comment on the usefulness of this result? Is the result close to the reality? I mean how many solutions are there by computers?}
}
\end{exa}

\vskip 0.3cm

The following corollary is obtained from Corollary 2.2.

\begin{cor}  {\rm Let $Q$ be the least prime factor of $q^m-1$ and $n=\frac{q^m-1}{Q}$. 
Let $\alpha$ and $r$ be two positive integers, and  $f(x)=x^{n+p^{\alpha}}+ax\in \Bbb F_{q^m}[x]$, where $a\in \Bbb F_{q^m}^*$ and $1< n+p^{\alpha}\le q^m-1$.  
Let $N_{f,q^m}$ be the number of solutions $(x,y)\in \Bbb F_{q^m}^2$ of an Artin-Schreier equation $y^q-y=f(x)$.  {If $x^{n}+a^{p^{\alpha}}$ has a solution  in $\Bbb F_{q^m}$}, then we get 
{
\begin{equation*}\bigg|N_{f,q^m}-q^m-\frac{(q-1)q^m}{Q}\bigg|\le (q-1)(Q-1)\sqrt{q^m};\end{equation*}
}
otherwise, we have 
\begin{equation*}\bigg|N_{f,q^m}-{q^m}\bigg|\le (q-1)Q\sqrt{q^m}.\end{equation*}
 
 }

 \end{cor}

  The following corollary is obtained from Proposition  1.7.

 \begin{cor} {\rm Let $s$ be an integer with $p<s<\sqrt {q^m}$ and $\gcd(s,p)=1$.  Let $f(x)=x^{sq^m/p}+ax^r\in \Bbb F_{q^m}[x]$, where $a\in \Bbb F_{q^m}^*$ and $r<s$.  Let $N_{f,q^m}$ be the number of solutions $(x,y)\in \Bbb F_{q^m}^2$ of an Artin-Schreier equation $y^q-y=f(x)$. Then 
\begin{equation*}\bigg|N_{f,q^m}-{q^m}\bigg|\le (q-1)(s-1)\sqrt{q^m}.\end{equation*}

}
\end{cor}

\subsection*{Acknowledgements}
The authors thank the anonymous  reviewers for their  helpful suggestions.

%DCC 

\end{document}